\documentclass[11pt]{article}
\usepackage[T1]{fontenc}

  \usepackage{amsthm}
\usepackage{amssymb}
\usepackage{amsmath}
\usepackage{longtable}
\usepackage{thm-restate}
\usepackage[inline]{enumitem}
\usepackage{tikz}
\usepackage{bbding}
\usetikzlibrary{arrows}
\usetikzlibrary{positioning}
\usetikzlibrary{fit}
\usetikzlibrary{decorations.pathmorphing}
\usetikzlibrary{patterns,patterns.meta,decorations.pathreplacing}
\tikzset{
    edge/.style={->,> = latex'}
}
\usepackage{hyperref} \usepackage[noabbrev]{cleveref}
\crefname{enumi}{property}{properties}
\crefname{corollary}{Corollary}{Corollaries}
  \usepackage[skip=0pt]{caption}
\usepackage{graphicx}
\usepackage{color}

\usepackage{algorithm}
\usepackage[noend]{algpseudocode}

\newcommand{\twfoarrow}[1]{\Rightarrow_{\text{TWFO}}^{#1}}

\newcommand{\disjun}{\uplus}

\newcommand{\sqsupsetsim}{\mathrel{\raisebox{0.5ex}{$\mathop{\kern0pt \sqsupset} \limits_{\textstyle\sim}$}}}
\newcommand{\sqsubsetsim}{\mathrel{\raisebox{0.5ex}{$\mathop{\kern0pt \sqsubset} \limits_{\textstyle\sim}$}}}

\newcommand{\vars}{\operatorname{vars}}

\newcommand{\comp}{\operatorname{comp}}

\newcommand{\mgu}{\operatorname{mgu}}

\newcommand{\dom}{\operatorname{dom}}
\newcommand{\cdom}{\operatorname{codom}}

\newcommand{\shortrules}[6]{\noindent\begin{minipage}{#6ex}{\bfseries #1}\end{minipage} $\;$ #2 $\;\Rightarrow_{\text{#5}}\;$ #3 \par\smallskip\noindent #4}

\newcommand{\CDCLBacktrackStack}[7]{
\tikzset{data/.style={minimum width = 0.6cm, minimum height = 0.6cm, draw, inner sep = 0}}
\def\a{0}
\if#11 \edef\a{1} \else \ifx#12 \edef\a{2} \else \edef\a{0} \fi \fi
\node [data, label={[align=center]below:decision\\level}] (A) {\a};
\node [data, right=0.7cm of A, label={[align=center]below:control\\stack}] (B) {0};
\node [minimum width = 0.6cm, minimum height = 0.6cm, right = 0.7cm of B, label={[align=center]below:trail}] {};
\if#20 \node [data, above=0cm of B] (D) {0}; \else \if#21 \node [data, above=0cm of B] (D) {0}; \node [data, above=0cm of D] {1}; \fi \fi
\if#31 \node [data, right = 0.7cm of B] (C) {$P$}; \fi
\foreach \nb / \of / \name / \label in {{#4/C/E/$\neg Q$}, {#5/E/F/$\neg T$}, {#6/F/G/$S$}, {#7/G//$\neg R$}} {
	\if\nb1 \node [data, above=0cm of \of] (\name) {\label}; \fi
}
}

  \newtheorem{lemma}{Lemma}
  
  \theoremstyle{definition}
  \newtheorem{definition}[lemma]{Definition}
  \newtheorem{example}[lemma]{Example}
  \newcommand{\inst}[1]{\textsuperscript{#1}}
  \newcommand{\orcidID}[1]{}
  \newcommand{\authorrunning}[1]{}
  \newcommand{\titlerunning}[1]{}
  \newcommand{\email}[1]{\texttt{#1}}
  \newcommand{\institute}[1]{\date{\normalsize{\def\and{\\[0.3em]}#1}}}
  \newcommand{\keywords}[1]{\par\smallskip\noindent\textbf{Keywords: }{\def\and{; }#1}}
  \newcommand{\doi}[1]{\href{https://doi.org/#1}{\nolinkurl{https://doi.org/#1}}}

\makeatletter
\newenvironment{breakablealgorithm}
  {\begin{center}
     \refstepcounter{algorithm}\hrule height.8pt depth0pt \kern2pt\renewcommand{\caption}[2][\relax]{{\raggedright\textbf{\fname@algorithm~\thealgorithm} ##2\par}\ifx\relax##1\relax \addcontentsline{loa}{algorithm}{\protect\numberline{\thealgorithm}##2}\else \addcontentsline{loa}{algorithm}{\protect\numberline{\thealgorithm}##1}\fi
       \kern2pt\hrule\kern2pt
     }
  }{\kern2pt\hrule\relax \end{center}
  }
\makeatother

\begin{document}
\title{A Two-Watched Literal Scheme for First-Order Logic}
\author{
Yasmine Briefs\inst{1,2}\orcidID{0009-0007-9917-7517}
\and Martin Bromberger\inst{1}\orcidID{0000-0001-7256-2190}
\and Tobias Gehl\inst{1}
\and Lorenz Leutgeb\inst{1,2}\orcidID{0000-0003-0391-3430}
\and Simon Schwarz\inst{1,2}\orcidID{0000-0002-0064-2922}
\and Christoph Weidenbach\inst{1}\orcidID{0000-0001-6002-0458}}
\authorrunning{Y.~Briefs et al.}
\institute{Max Planck Institute for Informatics, Saarbrücken, Germany\\
\email{\{ybriefs,mbromber,tgehl,lorenz,sschwarz,weidenbach\}@mpi-inf.mpg.de} \and
Graduate School of Computer Science,\\Saarland Informatics Campus, Saarbrücken, Germany}

\maketitle

\begin{abstract}
\noindent
The two-watched literal scheme, a core component of efficient CDCL (Conflict-Driven Clause Learning) implementations for propositional logic, is extended to first-order logic.
Given a set of first-order clauses and a set of ground literals, our lifted two-watched literal scheme
efficiently detects all propagating and false clauses with respect to the ground literals.
We present the algorithm as a system of rules and prove its soundness and completeness.
Additionally, we provide an implementation of the two-watched literal scheme, which
outperforms a standard dynamic programming approach for detecting
propagatable literals and conflicts, especially when dealing with long clauses.

\end{abstract}

\section{Introduction}\label{section:introduction}

The two-watched literal scheme is one key technology for the efficiency of modern CDCL SAT solvers \cite{DBLP:conf/iccad/SilvaS96,DBLP:conf/dac/MoskewiczMZZM01}.
The solvers build a \emph{trail} $M$, a consistent sequence of propositional literals, representing a
partial model for the considered set of clauses $N$. Then the clauses in $N$ are checked for
propagations or conflicts with respect to $M$. If a clause is false in $M$ it is a conflict.
If it is false for all literals except for one literal  undefined in $M$, it propagates this literal.
Clearly, watching the status of two different literals in any clause with respect to $M$
is sufficient to detect propagations and conflicts. While $M$ changes, these literals have to
be updated satisfying in particular the invariant: if one of the watched literals is false
and the other watched literal is not true,
then  all other literals of the clause are false. The two-watched literal scheme saves updates
on clauses while $M$ is extended and no update is needed if $M$ is shrunk.
Obviously, the savings depend on the length of clauses. For short clauses, e.g., clauses
of length two, the two-watched literal scheme has no advantage over standard algorithms detecting propagation
and conflict. Modern SAT solvers typically implement separate mechanisms for short clauses, in
particular for clauses of length one and two \cite{sorensson2005minisat}.

We have already suggested a lifting of the two-watched literal scheme to first-order logic in a workshop~\cite{DBLP:conf/paar/BrombergerGLW22}.
The inference system presented in Section~\ref{sec:calculus} is a refinement of this system. This paper
presents the first implementation of the first-order two-watched inference system and experimentally compares it with a standard dynamic
programming algorithm detecting propagations and conflicts.

Detecting propagations and conflicts in first-order logic with respect to a consistent sequence of ground literals, a (first-order) \emph{trail}, is
more complicated than in propositional logic, it even is NP-complete (\Cref{proposition:conflictNPhard}). This complexity comes from universally quantified variables in clauses. While the trail of ground
literals corresponds to a propositional trail, the detection of propagations 
and conflicts gets far more complex. For example, the first-order clause\newline
\centerline{$R(a,x) \lor R(y,x) \lor R(y,b)$}
already propagates the literal $R(a,b)$ with respect to an empty trail $M$. So in contrast to propagation in a SAT
setting, due to factorization, we may obtain propagations even though not all literals in the clause apart from the propagated literal
are falsified by the trail.
In addition, due to the existence of variables, a clause can disseminate more than one instance; actually, one
instance for every ``different'' instantiation. We consider an instantiation to be different if the actual
instantiating substitutions are different for two instances of the same clause.

We show that our two-watched literal scheme outperforms a dynamic programming approach to the
detection of false clauses and propagatable literals. To this end, we use the ground trail model building
by SCL(FOL)~\cite{BSW23SCL,BrombergerEtAl23CADE}. Whenever a new literal enters the trail, we independently call the implemented
two-watched literal scheme and the dynamic programming algorithm and compare the runtime and the number of
generated clause instances.

The paper is structured as follows: \Cref{sec:prelim} introduces the concepts and notations used in this paper. In \Cref{sec:problem}, we formally state the problem of incremental propagation and conflict detection and show NP-completeness.
\Cref{sec:naive} presents a baseline dynamic programming algorithm that solves the problem. In \Cref{sec:calculus}, we present rules of the two-watched literal scheme. This section also contains all relevant correctness proofs.
\Cref{section:evaluation} then shows empirical efficiency results of both approaches. Last, \Cref{section:conclusion} concludes the paper.

This is an extended version of a paper accepted at IJCAR 2026 \cite{briefs_2026_ijcar}
The implementation artifacts used in the evaluation can be found on Zenodo \cite{briefs_2026_20138614}.

\section{Preliminaries} \label{sec:prelim}

We assume a first-order language without equality where
$N$ denotes a clause set;
$C, D$ denote clauses;
$L, K, H$ denote literals;
$A, B$ denote atoms;
$P, Q, R$ denote predicates;
$t, s$ terms;
$f, g, h$ function symbols;
$a, b, c$ constants;
and $x, y, z$ variables.
Atoms, literals, clauses and clause sets are considered as usual, where
in particular clauses are identified both with their disjunction and the multiset
of their literals.
The complement of a literal is denoted by the function $\comp$.
Semantic entailment $\models$ is defined as usual where variables in clauses
are assumed to be universally quantified.
Substitutions $\sigma, \tau$ are total mappings from variables to terms, where
$\dom(\sigma) := \{x \mid x\sigma\neq x\}$ is finite and $\cdom(\sigma) := \{ t\mid x\sigma = t, x\in\dom(\sigma)\}$.
Their application is extended to literals, clauses, and sets of such objects in the usual way. In particular, we refer to a clause $C\sigma$ as a \emph{clause instance} of $C$ and call a literal $L\sigma$ an instance of $L$.
The function $\vars(Z)$ returns the set of variables in a term/atom/literal/clause/clause set $Z$.
A term, atom, clause, or a set of these objects is \emph{ground} if it does not contain any variable.
A substitution $\sigma$ is \emph{ground} if $\cdom(\sigma)$ is ground. A substitution $\sigma$ is \emph{grounding}
for a term $t$, literal $L$, clause $C$ if $t\sigma$, $L\sigma$, $C\sigma$ are ground, respectively.
A substitution $\sigma$ is the \emph{minimally grounding substitution} for a term/atom/literal/clause $Z$ if $Z \sigma$ is ground and $\dom(\sigma) = \vars(Z)$.
The function mgu denotes the \emph{most general unifier} of two terms, atoms, literals.
We assume that any mgu of two terms or literals does not introduce any fresh variables and is idempotent. In this work, we never compute unifiers between different clauses. Thus, we do not require different clauses to be variable disjoint.

\section{First-order Propagation \& Conflict Detection}

This section first states the formal problem of (incremental) first-order propagation and conflict detection, and classifies this problem as NP-complete. \Cref{sec:naive} then presents our baseline dynamic programming algorithm. In \Cref{sec:calculus}, we present our main result, the two-watched literal scheme for first-order logic, prove its correctness, and give an example run.

\subsection{Problem Statement} \label{sec:problem}

Formally, we consider the following problem: The algorithm is given a set of first-order clauses $N$, and a \emph{trail} $M$. The trail is a consistent sequence of ground literals. Literals can be appended or removed from the end of the trail in a classical stack-like fashion. In our setting, the trail forms a ground partial model assumption for $N$. Given $N$ and $M$, the algorithm needs to recognize and handle two scenarios:
\begin{itemize}
	\item If the clause set $N$ is false with respect to the partial model assumption, the algorithm should detect this and find a falsified clause.
	\item If there is no conflict, the algorithm should compute any propagatable literals from $N$ under $M$.
\end{itemize}
To this end, \Cref{def:canprop} formally defines what it means that a clause or a clause instance can propagate a literal with respect to a trail.
\begin{definition} \label{def:canprop}
A clause $C$ \emph{propagates} the literal $L\sigma$  with respect to a trail $M$ if $C\sigma = C'\sigma \lor L\sigma \lor \ldots \lor L\sigma$, $L\sigma\not\in C'$,
$L\sigma$ is undefined with respect to $M$ and $M \models \neg C'\sigma$. A clause $C$ is a \emph{conflict} with respect to a trail $M$ if there exists a grounding substitution $\sigma$ such that $M\models \neg C\sigma$.
\end{definition}

\noindent
Please note that the above definition enables the propagation of non-ground literals.
The overall problem is not considered in an isolated setting. Instead, we are interested in an incremental setting,
where dynamic changes to the trail $M$ and the clause set $N$ happen. The following three operations are of particular interest:
\begin{itemize}
	\item A single new ground literal is appended to the trail $M$.
	This, for example, corresponds to standard operations in model building calculi like CDCL or  SCL where decisions or propagations are added to the trail.
      \item A trail suffix of $M$ is removed. This corresponds to a backtracking step, where parts of a (conflicting) partial model assumption are deleted with respect to a learned clause.
        If the removed suffix equals the whole trail, this application also captures a restart of the integrating calculus.
	\item A new clause is added to the clause set $N$. This corresponds to learning a new clause in the integrating calculi.
\end{itemize}

\subsubsection{Problem Complexity}
In propositional logic, conflict and propagation detection can be solved in linear time for fixed (partial) models and clauses. In contrast, both problems, even in the non-incremental setting, are NP-complete for first-order logic, as shown in \Cref{proposition:conflictNPhard,corollary:propagationNPhard}.

\begin{restatable}{lemma}{propositionConflictNPhard}\label{proposition:conflictNPhard}
	The following problem is NP-complete: Given a trail $M$ and a clause $C$, determine if there is a grounding $\sigma$ such that $C\sigma$ is false in $M$.
\end{restatable}
\newcommand{\propositionConflictNPhardProof}{
\begin{proof}
	Inclusion in NP is straightforward by guessing $\sigma$ and checking ground falsification, as the size of $\sigma$ is bounded linearly in $|M|$.
	The hardness is proven by reduction from 3SAT.
	Assume we are given propositional clauses $C_1,\dots,C_n$ with three literals each in the propositional variables $X_1,\dots,X_m$.
	Without loss of generality, each of these variables occurs in at least one clause.
	The first-order signature consists of the ternary predicates $P_1,\dots,P_n$ and the constants $a$ (which stands for true) and $b$ (which stands for false).
	We now construct a clause $C$ in the pairwise distinct variables $x_1,\dots,x_m$.
	For each propositional clause $C_l=(\neg)X_i\lor(\neg)X_j\lor(\neg)X_k$, with $l\in\{1,\dots,n\}$ and $i,j,k\in\{1,\dots,m\}$, we add the first-order literal $\neg P_l(x_i,x_j,x_k)$ to the clause $C$.
	Further, we define the following set of literals: $S_l:=\{P_l(x, y, z) \mid x, y, z \in\{a, b\}\}$.
	Now, let $t_{l,1}=b$ if $X_i$ occurs positively in $C_l$, and $t_{l,1}=a$ if $X_i$ occurs negated.
	We define $t_{l,2}$ for $X_j$ and $t_{l,3}$ for $X_k$ analogously.
	Then, for the clause $C_l$, we add all literals from the set $S_l\setminus\{P_l(t_{l,1},t_{l,2},t_{l,3})\}$ to the trail $M$.

	Now we prove that there exists a grounding $\sigma$ such that $C\sigma$ is false in $M$ if and only if there is a valuation $v$ that satisfies $C_1\land\dots\land C_n$.
	First, assume that there exists a grounding $\sigma$ such that $C\sigma$ is false in $M$.
	Clearly, since we assumed that each variable occurs in at least one clause, each of the variables $x_1,\dots,x_m$ must be mapped to $a$ or $b$ by $\sigma$.
	We define a valuation $v$ as $X_i\mapsto\texttt{true}$ if $\sigma(x_i)=a$ and $X_i\mapsto\texttt{false}$ if $\sigma(x_i)=b$.
	Take a clause $C_l=(\neg)X_i\lor(\neg)X_j\lor(\neg)X_k$, then since $C\sigma$ is false in $M$, $P_l(x_i,x_j,x_k)\sigma$ must occur in $M$.
	As we removed the only valuation that would make $C_l$ false from $S_l$, it follows that $v$ satisfies $C_l$.
	Now assume that there exists a valuation $v$ that satisfies $C_1\land\dots\land C_n$.
	We define a substitution $\sigma$ as $\sigma(x_i)=a$ if $v(X_i)=\texttt{true}$ and $\sigma(x_i)=b$ otherwise.
	By an analogous argument, it follows that $C\sigma$ is false in $M$.
\end{proof}
}
\propositionConflictNPhardProof

\begin{restatable}{corollary}{corollaryPropagationNPhard}\label{corollary:propagationNPhard}
	The following problem is NP-complete: Given a trail $M$ and clause $C$, determine if there is any literal that $C$ can propagate with respect to $M$.
\end{restatable}
\newcommand{\corollaryPropagationNPhardProof}{
\begin{proof}
	Inclusion is again straightforward by guessing the literal and the substitution under which the literal propagates, and checking for ground propagation.
	Hardness is shown by reduction from the problem from \Cref{proposition:conflictNPhard}.
	Assume we are given a clause $C$ and a trail $M$.
	Let $Q$ be a fresh predicate.
	Then $C\lor Q$ propagates the literal $Q$ with respect to $M$ if and only if $C$ has an instance that is false with respect to $M$. $C\lor Q$ cannot propagate any other literal.
\end{proof}
}
\corollaryPropagationNPhardProof

\subsection{Baseline: A Dynamic Programming Algorithm} \label{sec:naive}

At the core of our baseline dynamic programming algorithm is the function CheckClause. It checks a single clause $C$ for propagations or conflicts under a given model $M$.

\begin{breakablealgorithm}
	\caption{Check a given clause $C$ for propagations and conflicts under $M$}\label{algorithm:checkClauseNaive}
	\begin{algorithmic}[1]
	\item[] \textbf{Input:} $M$ is a trail, $C$ is the clause to check, $K$ is either \texttt{NULL} or a literal representing a propagation and $\sigma$ defines the instance of $C$ that is currently considered
	\item[] \textbf{Output:} Returns \texttt{Conflict} if $C(\lor K)\sigma$ has a false instance under $M$, otherwise adds all propagations of $C(\lor K)\sigma$
		\Statex
		\Function{CheckClause}{$M,C,K,\sigma$}
			\If{$C = \emptyset$}
				\If{$K=\texttt{NULL}$}
					\State\Return{$\texttt{Conflict}$}
				\EndIf
				\If{$(K\sigma)$ is non-ground or not defined in $M$}
					\State{add ($K\sigma$) to the set of propagations}
				\EndIf
				\State\Return{\texttt{No Conflict}}
			\EndIf
			\State $C = \{ L \} \uplus C'$
			\If{$L\sigma$ is ground and true in $M$}
				\State\Return{\texttt{No Conflict}}
			\EndIf
			\State $\textit{insts}\gets\text{ instances of $\comp(L\sigma)$ in $M$ (using discrimination tree)}$
			\For{$L'$ in \textit{insts}}
				\If{\Call{CheckClause}{$M,C',K,\sigma\mgu(L\sigma,L')$} $=$ \texttt{Conflict}}
					\State\Return{\texttt{Conflict}}
				\EndIf
			\EndFor
			\If{$K=\texttt{NULL}$}
				\State\Return{\Call{CheckClause}{$M,C',L,\sigma$}}
			\EndIf
			\If{$L\sigma$ is unifiable with $K\sigma$}
				\State\Return{\Call{CheckClause}{$M,C',K,\sigma\mgu(L\sigma,K\sigma)$}}
			\EndIf
			\State\Return{\texttt{No Conflict}}
		\EndFunction
	\end{algorithmic}
\end{breakablealgorithm}

The check in line 9 serves to avoid unnecessary recursion; a true clause cannot propagate or be false.
The loop in lines 12 to 14 recurses over all possibilities to make the current literal $L\sigma$ false.
Finally, in lines 15 to 18, the algorithm tries to propagate the current literal $L\sigma$.
If there is no propagation yet, i.e., $K=\texttt{NULL}$, it can simply recurse with $L$ as the propagation, but if $K\neq\texttt{NULL}$, it can only recurse if $L\sigma$ and $K\sigma$ are unifiable.

To use the function CheckClause in an incremental setting, the algorithm maintains as its state a list of propagations together with the indices of the respective literals that made the propagations possible.
When a new literal is added to the trail, all clauses that contain generalizations of this literal are determined using a discrimination tree.
Then, CheckClause is called for each of these clauses with $\sigma$ being the matcher from the clause literal to the trail literal.
When a new clause is added, it is additionally necessary to check if all of its literals are unifiable, and if so, add a unit propagation.
When a suffix of the trail is removed, all that has to be done is to remove all propagations from the state that only became possible due to literals in the deleted suffix.

The running time of the function CheckClause can be estimated as $\mathcal{O}((1+|M|)^{|C|})$, which gives a running time of $\mathcal{O}(n\cdot(1+|M|)^{m})$ where $n$ is the number of clauses and $m$ is the maximal number of literals in a clause for the addition of a new literal to the trail.

\subsection{A Two-watched Literal Scheme} \label{sec:calculus}

We give the two-watched literal scheme for first-order logic in the form of a rule-based calculus. First, we give rules that find conflicts and propagations with respect to a fixed trail $M$. We prove soundness and completeness for this calculus.
Second, we consider the incremental setting described in \Cref{sec:problem}. The incremental operations are realized via parametrized \emph{steering} rules in our calculus. Concretely, these operations are modelled by the three rules AddLiteral($L$), RemoveLiteral, and AddClause($C$).
The rule AddLiteral($L$) allows adding a new ground literal $L$ to the trail.
Conversely, the rule RemoveLiteral removes a single literal from the trail. Multiple subsequent applications allow removing a trail suffix of arbitrary length. Last, the rule AddClause($C$) adds a clause $C$ to the considered clause set $N$.

\subsubsection{The Calculus State}
The calculus works on a tuple $(M; O; F; D)$. The state consists of the following components:
\begin{itemize}
	\item[$M$] is a trail consisting of consistent ground literals, representing the current partial model assumption.
	\item[$O$] is a set of tuples $(C; L_1; L_2)$ where $C$ is a clause instance and
		$L_1$ and $L_2$ are the two literals that are watched in $C$.
		A clause instance $(C; L_1; L_2) \in O$ can be annotated with $\top$, i.e. $(C; L_1; L_2)^{\top}$,
		which means that $C$ is an initial or learned clause, i.e. a not instantiated clause instance. All other clause instances are not annotated.
                Instances are never removed.
              \item[$F$] is a set of annotated literals $L^{C;K}$ with respect to $M$. The literal $K$ is the leftmost literal
                in $M=M'KM''$ such that $L$ propagates from $M'K$ and $C$.
		The literal $K$ is used for backtracking, i.e., the rule RemoveLiteral below.
		If $K=\top$ it can be propagated with an empty trail from $C$, i.e.,
		$C$ can be factorized to a unit clause.
	\item[$D$] is the conflict clause, or $\top$.
\end{itemize}

\noindent
A valid start state is $(M; O; F; \top)$ for a clause set $N$, trail $M$ that may be empty, and $O := \{(C; L_1; L_2)^\top \mid C\in N, \text{\ for some $L_1, L_2\in C$, $L_1 \not= L_2$ iff $|C|>1$}\}$.
We assume that initially propagatable units, potentially buildable by factoring, are all contained in $F$, and no non-propagatable literal is in $F$.
Further, we call a start state \emph{empty} if $M = \varepsilon$.

\subsubsection{The Calculus Rules}
In the following, we give the calculus rules formally in the form of a rewrite system on the above state.

The rule system computing propagations and conflicts consists of the \emph{internal rules} CreateInstance, UpdateWatched, FactorizeWatched, DetectPropLiteral, and Conflict. It does not depend on a particular way $M$ is created. Soundness, completeness, and termination are shown for any start state, including start states with a non-empty $M$.

Then, in  a second step, we introduce the \emph{steering rules} AddLiteral($L$), RemoveLiteral, and AddClause($C$), which manipulate the trail and clause set. We prove that, with a reasonable strategy and from an empty start state, rule applications of internal rules can be heavily restricted to consider only watched literals, or newly created instances built thereof.

We use the following terminology:
A literal $L$ is called \emph{true} if $L\in M$, \emph{false} if $\comp(L) \in M$, \emph{assigned} if it is either true or false, and \emph{unassigned} otherwise. Note that non-ground literals are always unassigned.
Further, a literal $L_j\sigma$ is \emph{fresh} in $(C; L_1; L_2)$ if $L_j\sigma\in C\sigma$, $L_j\sigma \not= L_1\sigma$, $L_j\sigma \not= L_2\sigma$.
In the whole following section, without loss of generality, the watched literals $L_1$ and $L_2$ should be considered interchangeably.
A state $(M; O; F; \top)$ is always implicitly associated with its clause set $N = \{ C \mid (C; L_1; L_2)^\top \in O \}$. The associated clause set $N$ will remain invariant under all rule applications, except if new clauses are added.
Last, note that in the following section, annotations to clauses in $O$ and $F$ are only explicitly written when they are required, and may be omitted.

\medskip

\shortrules{CreateInstance}{$(M; O \disjun \{ (C; L_1; L_2) \}; F; \top)$ \\ \hspace*{4.86em}}{$(M; O \cup \{ (C; L_1; L_2), (C\sigma; L_j\sigma; L_k\sigma) \}; F; \top)$}
{provided all of the following:
\begin{enumerate*}[label=(\roman*), itemjoin={,\quad}]
\item $\comp(L_1\sigma) \in M$ for some minimal grounding substitution $\sigma$
\item $L_j\sigma\in C\sigma$
\item if not all literals in $C\sigma$ are equal then $L_k\sigma \not= L_j\sigma$ 
\item if possible $L_j\sigma\in M$, otherwise if possible $L_j\sigma$ unassigned,
otherwise $M = M'_{1} \comp(L_j\sigma) M''_{1}$ and for all literals $L\sigma \in C\sigma$ holds $\comp(L\sigma) \not\in M''_{1}$ \label{createInstance:firstRightmost}
\item if $\comp(L_2\sigma) \not\in M$ then $L_k = L_2$ else $L_k\sigma \in C\sigma$
and if possible $L_k\sigma \in M$, otherwise if possible $L_k\sigma$ unassigned,
and otherwise $M = M'_{2} \comp(L_k\sigma) M''_{2}$ and for all literals $L\sigma \in C\sigma \setminus \{ L_j\sigma \}$ it holds that $\comp(L\sigma) \not\in M''_{2}$ \label{createInstance:secondRightmost}
\item there exist no $K_1\sigma, K_2\sigma \in C\sigma$ with $(C\sigma; K_1\sigma; K_2\sigma) \in O$.
\end{enumerate*}
}{TWFO}{20}

\medskip\noindent
When an instance of a non-ground watched literal is false w.r.t. $M$, CreateInstance adds the respective instance to $O$ and picks appropriate watched literals.

\medskip
\shortrules{UpdateWatched}{$(M; O \disjun \{ (C; L_1; L_2) \}; F; \top)$ \\ \hspace*{4.86em}}{$(M; O \cup \{ (C; L_j; L_2)\}; F; \top)$}
{provided $\comp(L_1) \in M$, and
$L_2\not\in M$, and
fresh $L_j\in C$ where $L_j\in M$ or if no such $L_j$ exists $L_j$ is unassigned.}{TWFO}{20}

\medskip\noindent
UpdateWatched exchanges a false watched literal as long as the other watched is not true. It is replaced by a true literal or, if such
a literal does not exist, by an unassigned literal. For newly created instances, this is already handled by CreateInstance, but
for existing instances, UpdateWatched takes care of watched literals.
For example, consider the watched clause $(P(x)\lor Q(y)\lor R(z);P(x);Q(y))\in O$ with $M=\neg P(a)$.
UpdateWatched is not applicable to this clause, instead, CreateInstance has to be applied to add the instance $(P(a)\lor Q(y)\lor R(z);R(z);Q(y))$ to $O$.
If the watched clause was $(P(a)\lor Q(y)\lor R(z);P(a);Q(y))$ instead, then CreateInstance would not be applicable because this would create a duplicate instance, and UpdateWatched would have to be applied.

\medskip
\shortrules{FactorizeWatched}{$(M; O \disjun \{ (C; L_1; L_2) \}; F; \top)$ \\ \hspace*{4.86em}}{$(M; O \cup \{ (C; L_1; L_2), (C\sigma; L_j\sigma; L_k\sigma)  \}; F; \top)$}
{provided all of the following:
\begin{enumerate*}[label=(\roman*), itemjoin={,\quad}]
\item $L_1\sigma = L_2\sigma$ for mgu $\sigma$
\item fresh $L_j\sigma\in C\sigma$
\item $L_j\sigma \not= L_k\sigma$
\item if possible $L_j\sigma\in M$, otherwise if possible $L_j\sigma$ unassigned,
otherwise $M = M'_{1} \comp(L_j\sigma) M''_{1}$ and for all literals $L\sigma \in C\sigma$ holds $\comp(L\sigma) \not\in M''_{1}$ \label{factorizeWatched:firstRightmost}

\item if $\comp(L_2\sigma) \not\in M$ then $L_k = L_2$ else $L_k\sigma \in C\sigma$
and if possible $L_k\sigma \in M$, otherwise if possible $L_k\sigma$ unassigned,
and otherwise $M = M'_{2} \comp(L_k\sigma) M''_{2}$ and for all literals $L\sigma \in C\sigma \setminus \{ L_j\sigma \}$ it holds that $\comp(L\sigma) \not\in M''_{2}$ \label{factorizeWatched:secondRightmost}

\item there exist no $K_1\sigma, K_2\sigma \in C\sigma$ with $(C\sigma; K_1\sigma; K_2\sigma) \in O$.
\end{enumerate*}
}{TWFO}{20}

\noindent
If the two watched literals are unifiable, FactorizeWatched adds this instance to $O$ and picks appropriate watched literals.

\medskip
\shortrules{DetectPropLiteral}{$(M; O \disjun \{ (C; L_1; L_2) \}; F; \top)$ \\ \hspace*{4.86em}}{$(M; O \cup \{ (C; L_1; L_2) \}; F \cup \{ L_2^{C;K} \}; \top)$}
{provided all of the following:
\begin{enumerate*}[label=(\roman*), itemjoin={,\quad}]
\item $\comp(L_1) \in M$
\item $C = C_0 \lor L_2 \lor \ldots \lor L_2$
\item $L_2$ unassigned
\item $M \models \neg C_0$
\item $M = M' K M''$ and $\comp(K) \in C$, and for all $L \in C$ holds $\comp(L) \not\in M''$
\item there is no $H^{D;K'} \in F$ such that $H\tau = L_2$ for some substitution $\tau$ and $K'\in M' K$.
\end{enumerate*}
}{TWFO}{20}

\medskip\noindent
DetectPropLiteral adds newly found propagations to $F$: It adds the propagation $L_2$ to $F$ if $L_2$ is not an instance of another propagation already contained in $F$.
The rule implicitly factorizes on the unassigned literal.
Recall that we annotate the found propagatable literal with the last literal on the trail that falsifies a literal in the clause. Eventually, this will allow us to remove propagations when literals from $M$ are removed.

\medskip
\shortrules{Conflict}{$(M; O \disjun \{ (C; L_1; L_2) \}; F; \top)$ \\ \hspace*{4.86em}}{$(M; O \cup \{ (C; L_1; L_2) \}; F; C)$}
{provided $\comp(L_1) \in M$ and $\comp(L_2) \in M$, and $M\models \neg C$.}{TWFO}{20}

\medskip\noindent
Finally, Conflict detects if the given trail $M$ falsifies a clause instance.

\subsubsection{Properties} For the above internal rules, we prove termination, soundness and completeness. Note that the proofs are independent of the way $M$ is created.

\begin{restatable}[Termination]{lemma}{lemmaTermination}
  \label{lem:termination}
  Exhaustive application of the internal rules CreateInstance, UpdateWatched, FactorizeWatched, DetectPropLiteral, and
  Conflict from any start state terminates.
\end{restatable}
\begin{proof}
  Let $(M; O; F; D)$ be the start state. None of the rules changes $M$ and none of the rules changes
  $N = \{C\mid (C,L_1,L_2)^\top\in O \}$.
  The only rules that manipulate $O$ are CreateInstance and FactorizeWatched. Let $n$ be the maximal length of a clause in $N$.
  Then CreateInstance can at most be applied $|N| \cdot (|M|+1)^n$  times starting from $(M; O; F; D)$ by any application order of the rules. Recall CreateInstance
  only introduces new instances. For some clause $C$ there are worst case $2^{|C|}$ different factors.
  FactorizeWatched can at most be applied $(|N| \cdot (|M|+1)^n + |N|) \cdot 2^n$ times starting from $(M; O; F; D)$ by any application order of the rules.
  Recall FactorizeWatched only introduces new instances. Finally, the rules UpdateWatched, FactorizeWatched, DetectPropLiteral, and
  Conflict can only be applied finitely often once the set $O$ does not change anymore.
\end{proof}

\noindent
\Cref{lem:invexhaust} states the most important invariants about the watched literals. Those invariants are always established if all of the above rules are applied exhaustively.
Overall, those invariants guarantee that the propagation and conflict information for a given clause instance can be entirely determined by just looking at the watched literals:
If and only if one watched literal is false, the clause instance can propagate, and if and only if both watched literals are false, the overall clause instance is false.

\begin{restatable}[Two-Watched Invariants]{lemma}{lemmaInvexhaust} \label{lem:invexhaust}
Let $(M; O; F; D)$ be a state reached from a start state during a run of the calculus.
If the rules CreateInstance, UpdateWatched and FactorizeWatched are not applicable to this state, then the following properties hold:

\begin{enumerate}
\item For every instance $(C; L_1; L_2) \in O$ where there are two literals $K_1, K_2 \in C$ with $K_1 \not= K_2$ that are unassigned
with respect to $M$ it holds that either $L_1$ or $L_2$ is true or $L_1$ and $L_2$ are unassigned with respect to $M$. \label{lem:invexhaust:exhaustbothunassign}

\item For every instance $(C; L_1; L_2) \in O$ with $C = C' \lor L_1 \lor \ldots \lor L_1 \lor L_2 \lor \ldots \lor L_2$ and $L_1, L_2 \not\in C'$
where $L_1$ and $L_2$ are not true and $L_1$ or $L_2$ is false with respect to $M$ it holds that $M \models \neg C'$.  \label{lem:invexhaust:exhauststillfalse}

\item For every instance $(C; L_1; L_2) \in O$ where $L_1$ and $L_2$ are false with respect to $M$ it holds that $M \models \neg C$.  \label{lem:invexhaust:exhaustbothfalse}
\end{enumerate}
\end{restatable}

\begin{proof} \medskip\noindent \ref{lem:invexhaust:exhaustbothunassign}. Let $(M; O; F; \top)$ be a state reached during a run of the calculus
where CreateInstance, UpdateWatched and FactorizeWatched are exhaustively applied. Further, we assume there is an instance $(C; L_1; L_2) \in O$ with two literals $K_1, K_2 \in C$ and $K_1 \not= K_2$ that are unassigned with respect to $M$.
For the sake of a proof by contradiction, we assume that $L_1$ is false with respect to $M$ and $L_2$ not true with respect to $M$.
We know that $L_1, L_2 \in C$ and $L_1 \not= L_2$ because there are at least two different literals in $C$,
because $K_1$ and $K_2$ are different.
Because $L_1$ is false it holds that $L_1 \not= K_1$ and $L_1 \not= K_2$
and because $K_1 \not= K_2$ it holds that $L_2 \not= K_1$ or $L_2 \not= K_2$.
So either $K_1$ or $K_2$ is fresh in $C$ and $\comp(L_1) \in M$ because $L_1$ is false
and $L_2 \not\in M$ because $L_2$ is not true.
So all prerequisites of the rule UpdateWatched are fulfilled
and therefore UpdateWatched is applicable contradicting our assumption.

\medskip\noindent \ref{lem:invexhaust:exhauststillfalse}. Let $(M; O; F; \top)$ be a state reached during a run of the calculus
where CreateInstance, UpdateWatched and FactorizeWatched are exhaustively applied.
Further let $(C; L_1; L_2) \in O$ be an instance with $C = C' \lor L_1 \lor \ldots \lor L_1 \lor L_2 \lor \ldots \lor L_2$ and $L_1, L_2 \not\in C'$
where, without loss of generality, $L_1$ is false and $L_2$ is not true with respect to $M$.
Assume that $M \not\models \neg C'$.
So there is at least one literal $L \in C'$ that is not false, hence it is either true or unassigned.
So $\comp(L_1) \in M$ holds as well as $L_2 \not\in M$ and $L$ is fresh in $C$ and either true or unassigned.
So UpdateWatched would be applicable contradicting our assumption.
So $M \models \neg C'$ holds.

\medskip\noindent \ref{lem:invexhaust:exhaustbothfalse}. The proof is analogous to the proof of \ref{lem:invexhaust:exhauststillfalse}
because $C = C' \lor L_1 \lor \ldots \lor L_1 \lor L_2 \lor \ldots \lor L_2$ holds and therefore $M \models \neg C$ implies $M \models \neg C'$.
\end{proof}

\noindent
The following lemma will establish some invariants about the TWFO calculus that will be later used in proofs.

\begin{restatable}[Invariants]{lemma}{lemmaInvariants} \label{lem:inv}
For a state $(M; O; F; D)$ that is reached from the start state during a run of the calculus, the following invariants hold:
\begin{enumerate}
\item Every instance $(D; K_1; K_2) \in O$ is an instance of an initial or added clause instance $(C; L_1; L_2)^{\top} \in O$. \label{lem:inv:actualinstance}

\item For every instance $(C; L_1; L_2) \in O$ it holds that $L_1 \in C$ and $L_2 \in C$. \label{lem:inv:watchedinclause}

\item For a clause instance $(C; L_1; L_2) \in O$ it holds that either
all literals in $C$ are equal or $L_1 \not= L_2$. \label{lem:inv:watchedunequal}

\item For all literals $L^{C;\top} \in F$ holds that $C = L \lor \ldots \lor L$. \label{lem:inv:proptopunit}

\item For all literals $L^{C;K} \in F$ with $K \not= \top$ holds that $M = M' K M''$
and $C = C_0 \lor L \lor \ldots \lor L$ and $M' K \models \neg C_0$ and $M' \not\models \neg C_0$. \label{lem:inv:propkontrail}

\item For a clause instance $(C; L_1; L_2) \in O$ where $L_1$ is ground and
UpdateWatched is exhaustively applied
it holds that $C$ is true or $L_1$ is unassigned or the instance generates a
conflict or can propagate or the literal that would be propagated is an instance of a literal in $F$ or already true in $M$. \label{lem:inv:exhaustgroundfalseimpl}
\end{enumerate}

\end{restatable}
\begin{proof}
\noindent \ref{lem:inv:actualinstance}. We prove this by induction.
For the start state the invariant holds
because all instances are just the initial clauses.
The induction hypothesis is that the invariant holds for a state $(M; O; F; D')$.
For every instance $(D; K_1; K_2) \in O$ there is no rule that changes the clause instance $D$.
This means that if an instance $(D; K_1; K_2) \in O$ is an instance of an initial or added clause $C$ when it is created
then it will always be an instance of $C$, as $C$ cannot be removed.
So left to show is that every rule that creates an instance creates an actual instance of an initial or added clause.
CreateInstance and FactorizeWatched are the only rules that create instances.
Both create an instance $(D\sigma; K_j\sigma; K_k\sigma)$ from an instance $(D; K_1; K_2) \in O$
where $\sigma$ is non-empty.
So $D\sigma$ is an instance of $D$, therefore $D\sigma$ is an instance of an initial or added clause $C$ if $D$ is, by transitivity of instantiation.
By induction follows that for all $(D; K_1; K_2) \in O$ holds that $D$ is an actual instance of an
initial or added clause $(C; L_1; L_2)^{\top} \in O$.

\medskip\noindent \ref{lem:inv:watchedinclause}. We prove this by induction.
The induction basis is the start state.
For the start state the invariant holds
because for all watched literals $L_1, L_2 \in C$ is required.
The induction hypothesis is that the invariant holds for a state $(M; O; F; D)$.
The rules CreateInstance, UpdateWatched, FactorizeWatched and AddClause are the only rules that change
the watched literals or add a new instance.
All these rules either do not change a watched literal $L_1$ or
they require for the newly watched literal $L_j$ that $L_j \in C$ holds.
So the invariant still holds after rule application.

\medskip\noindent \ref{lem:inv:watchedunequal}. We prove this by induction.
The induction basis is the start state.
For the start state the invariant holds
because $L_1 = L_2$ is only allowed for unit clauses $C$.
The induction hypothesis is that the invariant holds for a state $(M; O; F; D)$.
The rules CreateInstance, UpdateWatched, FactorizeWatched and AddClause are the only rules that change
the watched literals or add a new instance.
CreateInstance, FactorizeWatched and AddClause require $L_1 \not= L_2$ if possible.
So after application of these rules the invariant still holds.
UpdateWatched picks a fresh $L_j \in C$ and therefore $L_j \not= L_2$.
So after application of UpdateWatched the invariant still holds.

\medskip\noindent \ref{lem:inv:proptopunit}. We prove this by induction.
The induction basis is the start state.
For the start state the invariant holds because for all $L^C \in F$ holds
that they are annotated with $\top$ and that $C = L \lor \ldots \lor L$ is a unit.
The induction hypothesis is that the invariant holds for a state $(M; O; F; D)$.
DetectPropLiteral and AddClause are the only rules that add literals to $F$
and no rule changes literals or annotations in $F$.
DetectPropLiteral only adds $L^{C;K}$ to $F$ where $K \not= \top$.
So the invariant still holds after DetectPropLiteral was applied.
AddClause can add $L\sigma^{C\sigma;\top}$ to $F$ where $C\sigma = L\sigma \lor \ldots \lor L\sigma$.
So the invariant still holds after AddClause($C$) was applied.

\medskip\noindent \ref{lem:inv:propkontrail}. We prove this by induction.
The induction basis is the start state.
For the start state the invariant holds because for all $L^{C;K} \in F$ holds $K = \top$.
The induction hypothesis is that the invariant holds for a state $(M; O; F; D)$.
DetectPropLiteral, RemoveLiteral and AddClause are the only relevant rules that we have to consider
because all other rules just extend the trail or do not change $M$ and $F$.
AddClause can only add literals $L^{C;\top}$. So the invariant still holds after AddClause was applied.
DetectPropLiteral adds $L^{C;K}$ to $F$ where $C = C_0 \lor L \lor \ldots \lor L$ and $L$ is unassigned and $M \models \neg C_0$ hold.
Also $M = M' K M''$ and $\comp(K) \in C$ and for all $L' \in C$ holds that $\comp(L) \not\in M''$.
Because $M \models \neg C_0$ and for all $L' \in C$ it holds that $\comp(L) \not\in M''$ and it holds that $M' K \models \neg C_0$.
We know that $\comp(K) \not= L$ holds because $L$ is unassigned and therefore $\comp(K) \in C_0$.
So because $\comp(K) \in C_0$ and $M' K \models \neg C_0$ it also holds that $M' \not\models \neg C_0$.
So after DetectPropLiteral was applied the invariant still holds.
RemoveLiteral removes the last literal $L'$ from the trail and removes all literals from $F$
that are annotated with $L'$.
We have to show that all literals that remain in $F$ still fulfill the invariant for the shorter trail.
Let $M = M'_1 L'$. We know that for all $L^{C;K} \in F$ with $K \not= \top$ holds that $M = M' K M''$
and $C = C_0 \lor L \lor \ldots \lor L$ and $M' K \models \neg C_0$ and $M' \not\models \neg C_0$.
Only if $K \not= L'$ the literal $L$ is kept in $F$
and we know that $M = M'_1 L' = M' K M''' L'$ and therefore $M'_1 = M' K M'''$.
So for all literals still in $F$ after the application of RemoveLiteral the invariant still holds
and therefore the invariant still holds after application of RemoveLiteral.

\medskip\noindent \ref{lem:inv:exhaustgroundfalseimpl}. Let $(C; L_1; L_2) \in O$
be a clause instance where $L_1$ is ground and UpdateWatched is exhaustively applied.
Either $L_1$ is true with respect to $M$ and then $C$ is also true or $L_1$ is unassigned
or $L_1$ is false with respect to $M$.
Because UpdateWatched is exhaustively applied we know that if $L_1$ is false
then $L_2$ is not true and there is no other true or unassigned literal.
So all literals that are not watched are false.
If $L_2$ is false then Conflict is applicable.
If $L_2$ is unassigned then all prerequisites for the DetectPropLiteral rule are met except
that there is no $H \in F$ such that $H\tau = L_2$ for some $\tau$
and this would mean that $L_2$ is an instance of a literal in $F$.
So the invariant holds.
\end{proof}

\noindent
Last, \Cref{lem:prop} states soundness and completeness. It shows that, after exhaustive rule applications, the calculus correctly computes all possible propagations and correctly determines if $N$ is conflicting to the trail $M$.

\begin{restatable}[Soundness and Completeness]{lemma}{lemmaProp} \label{lem:prop}
Let $(M; O; F; \top)$ be a state reached from the start state during a run of the calculus.
If the rules CreateInstance, UpdateWatched and FactorizeWatched are not applicable to this state, then the following properties hold:
\begin{enumerate}
\item If there is an instance $(C; L_1; L_2)^{\top} \in O$ with ground instance $C\sigma$
such that $M \models \neg C\sigma$ then Conflict is applicable for $C\sigma$. \label{lem:prop:complconflict}

\item If there is an instance $(C; L_1; L_2)^{\top} \in O$ with $C\sigma = C'\sigma \lor L\sigma \lor \ldots \lor L\sigma$ where $C'\sigma$ is ground and $\sigma$ is minimal
and $M \models \neg C'\sigma$ and $L\sigma$ is unassigned with respect to $M$ and not an instance of a literal in $F$ then DetectPropLiteral is applicable for $L\sigma$. \label{lem:prop:complpropagate}

\item If there is a transition $(M; O; F; \top) \Rightarrow_{\text{TWFO}}^{\text{Conflict}} (M; O; F; D)$ then
there is a conflict $M \models \neg D$ where $D$ is a ground instance of a clause instance $(C; L_1; L_2)^{\top} \in O$. \label{lem:prop:soundconflict}

\item If there is a transition $(M; O; F; \top) \Rightarrow_{\text{TWFO}}^{\text{DetectPropLiteral}} (M; O; F \cup \{ L^D \}; \top)$ then
there is a propagation where $D = D' \lor L \lor \ldots \lor L$ and $M \models \neg D'$ and $D$ is an instance of a clause instance $(C; L_1; L_2)^{\top} \in O$
and $D'$ is ground. \label{lem:prop:soundPropLit}
\end{enumerate}
\end{restatable}
\begin{proof}
\noindent \ref{lem:prop:complconflict}. Let $(M; O; F; \top)$ be a state of an application of the calculus reached during a run
and $(C; K_1; K_2)^{\top} \in O$ an initial clause where there is a ground instance $C\sigma$ such that $M \models \neg C\sigma$.
Further let the rules CreateInstance, UpdateWatched and FactorizeWatched be exhaustively applied.
So there is an instance $(C\sigma_1; L_1\sigma_1; L_2\sigma_1) \in O$ such that there is a possibly empty $\sigma'_1$
with $\sigma_1\sigma'_1 = \sigma$. There is such an instance because the instance $(C; K_1; K_2)^{\top}$
itself is such an instance where $\sigma_1$ is empty.
Because there could be multiple of such instances we assume that $C\sigma_1$ is minimal
with respect to the number of variables.
So there is no $(C\sigma_1\sigma_2; L_1\sigma_1\sigma_2; L_2\sigma_1\sigma_2) \in O$ where $\sigma_2$ is not empty
and there is a $\sigma'_2$ such that $\sigma_1\sigma_2\sigma'_2 = \sigma$.

If $L_1\sigma_1$ is non-ground we know that $L_1\sigma_1\sigma'_1 = L_1\sigma \in C\sigma$ is ground
and false with respect to $M$ and therefore either CreateInstance would be applicable or $C\sigma_1$ is not minimal.
But CreateInstance is exhaustively applied so $C\sigma_1$ is not minimal but this contradicts our assumption.
Using the same argument, it follows that $L_2\sigma_1$ is ground.
This implies that $L_1\sigma_1=L_1\sigma\in C\sigma$, and therefore $L_1\sigma_1$ is false with respect to $M$.
The same holds for $L_2\sigma_1$.
Since UpdateWatched is not applicable, all literals in $C\sigma_1$ must be false.
It follows that $C\sigma_1$ is ground, so in particular, $\sigma_1=\sigma$ and Conflict is applicable with the claimed instance.

\medskip\noindent \ref{lem:prop:complpropagate}. Let $(M; O; F; \top)$ be a state of an application of the calculus reached during a run
and $(C; K_1; K_2)^{\top} \in O$ where there is an instance $C\sigma = C'\sigma \lor L^{1}\sigma \lor \ldots \lor L^{k}\sigma$
where $C'\sigma$ is a non-empty ground instance. Further $M \models \neg C'\sigma$ and $L^{1}\sigma = \ldots = L^{k}\sigma$
and $L^{1}\sigma$ is unassigned with respect to $M$ and not an instance of a literal in $F$.
Further let the rules CreateInstance, UpdateWatched and FactorizeWatched be exhaustively applied.

Now as in part~\ref{lem:prop:complconflict} of the proof we again assume that there is a minimal instance
$(C\sigma_1; L_1\sigma_1; L_2\sigma_1) \in O$ such that there is a possibly empty $\sigma'_1$ with $\sigma_1\sigma'_1 = \sigma$.
Now we have to do a case analysis on which literals are watched.
If both watched literals are in $C'\sigma_1$ we know that if one is ground UpdateWatched would be applicable
because $L^{1}\sigma_1$ is unassigned.
If one is non-ground then there is a false ground instance because $M \models \neg C'\sigma$ holds
and then either CreateInstance would be applicable or $C\sigma_1$ is not a minimal instance.

If both watched literals $L_1\sigma_1$ and $L_2\sigma_1$ are in $L^{1}\sigma_1 \lor \ldots \lor L^{k}\sigma_1$ then
$L_1\sigma_1 \not= L_2\sigma_1$ has to hold as we know from \Cref{lem:inv}.\ref{lem:inv:watchedunequal}.
But then $L_1\sigma_1$ and $L_2\sigma_1$ are unifiable because $L^{1}\sigma = \ldots = L^{k}\sigma$ and $\sigma_1\sigma'_1 = \sigma$ hold.
So FactorizeWatched would be applicable or $C\sigma_1$ is not a minimal instance.

The last case is that one watched literal is in $C'\sigma_1$ and one in $L^{1}\sigma_1 \lor \ldots \lor L^{k}\sigma_1$.
Without loss of generality, we assume that $L_1\sigma_1 \in C'\sigma_1$ and $L_2\sigma_1 = L^{1}\sigma_1$.
So we know that either $L_1\sigma_1$ is ground and therefore it is false.
Then either UpdateWatched would be applicable or DetectPropLiteral is applicable for $L^{1}\sigma_1$
because $L^{1}\sigma_1$ is not an instance of a literal in $F$.
If $L_1\sigma_1$ is non-ground we know that $L_1\sigma_1\sigma'_1$ is false because $M \models \neg C'\sigma$ and $\sigma = \sigma_1\sigma'_1$ hold.
So CreateInstance would be applicable.

So from the above we get that if the rules CreateInstance, UpdateWatched and FactorizeWatched are exhaustively applied
then the only option left is that DetectPropLiteral is applicable for $L^{1}\sigma$.

\medskip\noindent \ref{lem:prop:soundconflict}. Let $(M; O; F; \top) \Rightarrow_{\text{TWFO}}^{\text{Conflict}} (M; O; F; D)$ be a transition.
This means there is an instance $(D; K_1; K_2) \in O$.
From Lemma~\ref{lem:inv}.\ref{lem:inv:actualinstance} we know that $D$ is an actual instance of an instance $(C; L_1; L_2)^{\top} \in O$.
It is also required in the Conflict rule that $M \models \neg D$ and
because $M$ does only contain ground literals $D$ is ground.
So the property holds.

\medskip\noindent \ref{lem:prop:soundPropLit}. Let $(M; O; F; \top) \Rightarrow_{\text{TWFO}}^{\text{DetectPropLiteral}} (M; O; F \cup \{ L^D \}; \top)$ be a transition.
Then we know from the prerequisites of DetectPropLiteral that there is an instance $(D; L_1; L_2) \in O$
where $D = D_0 \lor L_2 \lor \ldots \lor L_2$ and $M \models \neg D_0\sigma$ and $L = L_2$.
From Lemma~\ref{lem:inv}.\ref{lem:inv:actualinstance} we know that $D$ is an actual instance of an instance $(C; K_1; K_2)^{\top} \in O$.
From $M \models \neg D_0$ we know that $D_0$ is ground because $M$ only contains ground literals.
So the property holds.
\end{proof}

\medskip

\noindent
Up until this point, if we apply all of the above internal rules exhaustively from any starting state, we can determine all conflicting and propagatable literals in $N$ under $M$. In the following part, we now dynamically change $M$ and $N$. In this dynamic setting, the two-watched invariants (\Cref{lem:invexhaust}) will allow us to efficiently detect false or propagating clauses.

\subsubsection{Incremental Updates}
The following three \emph{steering rules} allow modifying the trail $M$, or add clauses to the considered clause set $N$.

\medskip
\shortrules{AddLiteral($L$)}{$(M; O; F; \top)$}{$(M L; O; F; \top)$}
{provided $L$ is ground and unassigned in $M$.
}{TWFO}{20}

\medskip
\shortrules{RemoveLiteral}{$(M L; O; F; D)$}{$(M; O; F'; D)$}
{provided $F' \subseteq F$,
for all $L'^{C;K} \in F'$ it holds that $L \not= K$,
and for all $L'^{C;K} \in F \setminus F'$ it holds that $L = K$, and
$D \not= \top$.}{TWFO}{20}

\medskip
\shortrules{AddClause($C$)}{$(M; O; F; D)$}{$(M; O \cup \{ (C; L_1; L_2)^{\top} \}; F'; \top)$}
{provided all of the following:
\begin{enumerate*}[label=(\roman*), itemjoin={,\quad}]
\item $D \not= \top$
\item there is no $\tau$ and $M = M' M''$ where $M''$ is non-empty such that $C\tau$ can already propagate with respect to $M'$
\item $L_1, L_2 \in C$
\item $L_1 \not= L_2$ if $C$ contains different literals
\item $L_1$ is unassigned \label{crit:addclause:unass}
\item if possible $L_2$ unassigned, otherwise $M = M' \comp(L_2) M''$ and for all literals $L \in C$ holds $\comp(L) \not\in M''$ \label{addClause:secondRightmost}
\item if there is a minimal substitution $\tau$ such that $C\tau = L\tau \lor \ldots \lor L\tau$
then $F' = F \cup \{ L\tau^{C\tau;\top} \}$ else $F' = F$
\item $M \not\models \neg D$. \label{crit:addclause:confresolved}
\end{enumerate*}}{TWFO}{20}

\medskip\noindent
Note that clauses can only be added after a Conflict was detected. This captures classical conflict-driven learning strategies. Also note that criterion \ref{crit:addclause:unass} requires the added clause to not be directly false under $M$, and criterion \ref{crit:addclause:confresolved} requires the previous conflict to no longer exist.

\medskip\noindent
To preserve the properties of \Cref{lem:prop,lem:invexhaust}, we need to exhaustively apply all internal rules before updating the trail or clause set. This is captured by the following definition:
\begin{restatable}[Reasonable Strategy]{definition}{definitionReasonable}
A strategy is called \emph{reasonable} if the rules CreateInstance, UpdateWatched, FactorizeWatched, DetectPropLiteral are preferred over Conflict, and all those internal rules are preferred over the steering rules AddLiteral($L$), RemoveLiteral, and AddClause($C$).
\end{restatable}
\noindent
The definition ensures that, in a reasonable run, right before any steering rule is applied, all conditions from \Cref{lem:prop,lem:invexhaust} hold. Thus, at such a moment in the derivation, all propagation and conflict information is captured and can be extracted from the state.
Next, we demonstrate the rules of the calculus in action by presenting an example run.
\begin{example}
Let $N = \{ (1) \, P(x) \lor \neg Q(x) \lor R(x, y), (2) \, P(x) \lor Q(a), (3) \, P(a) \lor \neg R(x, b) \}$
be the initial clause set. Furthermore, consider $(4) \, P(x) \lor P(a)$ as a clause added during the run.
To simplify notation, instead of writing down tuples $(C; L_1; L_2)$, we annotate watched literals with $^*$.
Let $O_0=\{(1)\,P(x)\lor\neg Q(x)^*\lor R(x, y)^*, (2) \, P(x)^* \lor Q(a)^*, (3) \, P(a)^* \lor \neg R(x, b)^* \}$ be our starting set of clause instances.
The following is a reasonable run of the two-watched literal scheme calculus:

{\renewcommand{\arraystretch}{1.5}
\raggedright
\begin{longtable}{l l}

& $(\varepsilon;
O_0;
\emptyset; \top)$ \\

$\twfoarrow{\text{AddLiteral}(\neg P(a))}$ &
$(\neg P(a);
O_0;
\emptyset; \top)$ \\

$\twfoarrow{\text{CreateInstance}(2)}$ &
$(\neg P(a);
O_1 := O_0 \cup \{(2.1)\,P(a)^* \lor Q(a)^* \};
\emptyset; \top)$ \\

$\twfoarrow{\text{DetectPropLit}(2.1)}$ &
$(\neg P(a);
O_1;
F_1:=\{ Q(a)^{(2.1); \neg P(a)} \}; \top)$ \\

$\twfoarrow{\text{DetectPropLit}(3)}$ &
$(\neg P(a);
O_1;
F_2:= F_1\cup\{ \neg R(x, b)^{(3); \neg P(a)} \}; \top)$ \\

$\twfoarrow{\text{AddLiteral}(Q(a))}$ &
$(\neg P(a) Q(a);
	O_1;F_2;\top)$\\

$\twfoarrow{\text{CreateInstance}(1)}$ &
$(\neg P(a) Q(a);
	O_2:=O_1\cup\{(1.1)\,P(a)\lor\neg Q(a)^*\lor R(a,y)^*\};$\\[-.5ex]
	&\hfill$F_2;\top)$\\[-.5ex]

$\twfoarrow{\text{DetectPropLit}(1.1)}$ &
$(\neg P(a) Q(a);
O_2; F_3:=F_2\cup\{R(a,y)^{(1.1); Q(a)} \}; \top)$ \\

$\twfoarrow{\text{AddLiteral}(R(a,b))}$ &
	$(\neg P(a) Q(a) R(a,b); O_2; F_3;\top)$\\

$\twfoarrow{\text{CreateInstance}(3)}$ &
	$(\neg P(a) Q(a) R(a,b); O_3:=O_2\cup\{(3.1)\,P(a)^*\lor\neg R(a,b)^*\};$\\[-.5ex]
	&\hfill$F_3;\top)$\\[-.5ex]

$\twfoarrow{\text{Conflict}(3.1)}$ &
$(\neg P(a) Q(a) R(a,b);
O_3;F_3; P(a) \lor \neg R(a, b))$ \\

$\twfoarrow{\text{RemoveLiteral}\times3}$ &
$(\varepsilon;
O_3;
\emptyset; P(a) \lor \neg R(a, b))$ \\

$\twfoarrow{\text{AddClause}(4)}$ &
$(\varepsilon;
O_3 \cup \{(4) P(a)^* \lor P(x)^*\};
\{ P(a)^{P(a) \lor P(x); \top} \}; \top)$ \\
\end{longtable}

}

\noindent
In the last step, the propagatable literal $P(a)$ is already detected by the implicit factoring in AddClause.
Further, note that although the literal $\neg P(a)$ falsifies an instance of the literal $P(x)$ in the clause $(1)$, this clause does not have to be touched after AddLiteral$(\neg P(a))$ because $P(x)$ is not watched, as we show in \Cref{lem:efficientApplication}. This allows to omit checks which leads to the efficiency of the two-watched literal scheme.

\end{example}

We now establish an invariant that holds when using incremental trail building from an empty starting state.
Intuitively, \Cref{lem:watchedInvariants} assures that watched, assigned literals are ``rightmost'' in the trail sequence $M$. This invariant is important if literals are removed from $M$. In such a case, the watched literal will become unassigned before any non-watched literals can change their status. Effectively, this condition is sufficient to keep all invariants from \Cref{lem:invexhaust} automatically after applying RemoveLiteral. Hence, it is never required to change watched literals in such a situation, which in turn allows omitting rule application checks (\Cref{lem:efficientApplication}). 

\begin{restatable}[Watched Literals are Rightmost]{lemma}{lemmaWatchedInvariants} \label{lem:watchedInvariants}
A state $(M; O; F; D)$ that is reached by a reasonable run from the empty starting state satisfies all of the following properties:
\begin{enumerate}
\item For every instance $(C; L_1; L_2) \in O$ where $C = C' \lor L_1 \lor \ldots L_1 \lor L_2 \lor \ldots \lor L_2$ and $L_1$ is false and $L_2$ is unassigned with respect to $M$ it holds that
$M = M' \comp(L_1) M''$ where for all literals $L \in C'$ holds $\comp(L) \not\in M''$. \label{def:fastback:onefalse}

\item For every instance $(C; L_1; L_2) \in O$ where $C = C' \lor L_1 \lor \ldots L_1 \lor L_2 \lor \ldots \lor L_2$ and $L_1$ and $L_2$ are false with respect to $M$ it holds that
$M = M' \comp(L_1) M'' \comp(L_2) M'''$ where for all literals $L \in C'$ holds $\comp(L) \not\in M'' M'''$. \label{def:fastback:bothfalse}

\item For every instance $(C; L_1; L_2) \in O$ where $C = C' \lor L_1 \lor \ldots L_1 \lor L_2 \lor \ldots \lor L_2$ and $L_1$ is false and $L_2$ is true with respect to $M$ it holds that
either $M = M'_{1} L_2 M''_{1} \comp(L_1) M'''_{1}$ or $M = M'_{2} \comp(L_1) M''_{2} L_2 M'''_{2}$
where for all literals $L \in C'$ holds $\comp(L) \in M'_{2}$. \label{def:fastback:falsetrue}
\end{enumerate}
\end{restatable}

\begin{proof}
	The proof is by induction on the number of applied rules.
	For the empty start state, all three properties are obviously true because $M=\varepsilon$, so no literals can be false.
	The induction step is shown by a case distinction on the last applied rule.
	For readability, we do not mention instances that the respective rule does not modify unless the rule modifies the trail.
	DetectPropLiteral and Conflict do not change the instances or the trail, so if they are the last applied rule, all three statements follow from the induction hypothesis.
	For RemoveLiteral, all three statements obviously follow from the induction hypothesis as well.

	\medskip\noindent\ref{def:fastback:onefalse}.
	CreateInstance asserts this in its premise as \Cref{createInstance:secondRightmost}.
	For UpdateWatched, let $(\overline{C};\overline{L_j};\overline{L_2})$ be the added instance and $\overline{L_1}$ be the previously watched literal.
	Then, to satisfy the condition of the statement we want to prove, $\overline{L_2}$ must be false in $M$ and $\overline{L_j}$ must be undefined.
	Hence, for the previous state, both watched literals $\overline{L_1}$ and $\overline{L_2}$ must have been false, and the statement follows from \Cref{def:fastback:bothfalse} of the induction hypothesis.
	For FactorizeWatched, let $(\overline{C}\sigma;\overline{L_j}\sigma;\overline{L_k}\sigma)$ be the instance added from $(\overline{C};\overline{L_1};\overline{L_2})$.
	If $\overline{L_j}$ is false, then the statement is asserted in the premise of the rule as \Cref{factorizeWatched:firstRightmost}.
	If $\overline{L_k}$ is false, then it is asserted in the premise of the rule as \Cref{factorizeWatched:secondRightmost}.
	For AddLiteral, let $(C;L_1;L_2)\in O$ be an instance that satisfies the condition of the statement we want to prove.
	Assume there is a literal $L\in C'$ that is false in $M''$.
	By the induction hypothesis, it must hold that $L=\comp(\overline{L})$ where $\overline{L}$ is the literal that AddLiteral added.
	This implies the presence of two unassigned literals $L$ and $L_2$ in the previous state, and by \Cref{lem:invexhaust:exhaustbothunassign} of \Cref{lem:invexhaust}, AddLiteral cannot be applied to the previous state in a reasonable strategy.
	AddClause asserts this in its premise as \Cref{addClause:secondRightmost}.

	\medskip \noindent\ref{def:fastback:bothfalse}.
	CreateInstance asserts this in its premise as \Cref{createInstance:firstRightmost,createInstance:secondRightmost}.
	UpdateWatched cannot create such an instance.
	FactorizeWatched asserts this in its premise as \Cref{factorizeWatched:firstRightmost,factorizeWatched:secondRightmost}.
	For AddLiteral, let $(C;L_1;L_2)\in O$ be an instance that satisfies the condition of the statement we want to prove.
	Assume there is a literal $L\in C'$ that is false in $M''M'''$.
	Let $\overline{L}$ be the literal that AddLiteral added.
	If $L_2=\comp(\overline{L})$, then by \Cref{def:fastback:onefalse} of the induction hypothesis, such an $L\in C'$ cannot exist.
	Otherwise, $L_1$ and $L_2$ are also false in the previous state.
	By \Cref{def:fastback:bothfalse} of the induction hypothesis, it must hold that $L=\comp(\overline{L})$.
	This means that in the previous state, $L$ was unassigned, and in a reasonable run, AddLiteral is not applicable to the previous state by \Cref{lem:invexhaust:exhaustbothfalse} of \Cref{lem:invexhaust}.
	AddClause cannot add such an instance because one of the literals in the added instance is unassigned.

	\medskip \noindent\ref{def:fastback:falsetrue}.
	For CreateInstance, let $(\overline{C};\overline{L_1};\overline{L_2})$ be the instance that the instance in question was created from and let $\sigma$ be the used substitution.
	First assume that $\comp(\overline{L_2}\sigma)\notin M$.
	It follows that $L_2=\overline{L_2}\sigma$.
	In this case, the statement is asserted in the premise of the rule as \Cref{createInstance:firstRightmost}.
	Now assume that $\comp(\overline{L_2}\sigma)\in M$.
	Then, the statement is asserted in the premise of the rule as \Cref{createInstance:firstRightmost,createInstance:secondRightmost}.
	For UpdateWatched, in the previous state, both watched literals must have been false for the instance in question, and by \Cref{def:fastback:bothfalse} of the induction hypothesis, the statement follows.
	The argument for FactorizeWatched works analogously to the argument we used for CreateInstance.
	For AddLiteral, we only have to argue about the case where $\comp(L_1)$ occurs to the left of $L_2$ in $M$.
	If $L_2$ is the literal added by AddLiteral, then this follows from \Cref{lem:invexhaust:exhauststillfalse} of \Cref{lem:invexhaust} together with \Cref{def:fastback:onefalse} of the induction hypothesis.
	If $L_2$ is not the literal added by AddLiteral, then this trivially follows by the induction hypothesis.
	AddClause cannot add such an instance because one of the literals in the added instance is unassigned.
\end{proof}

\noindent
The next lemma characterizes which internal rule applications have to be considered after applying a steering rule.

\begin{lemma}[Incremental Handling of Trail Updates] \label{lem:applcheckwatched}
For a run with a reasonable strategy the following hold:
\begin{enumerate}
\item Consider a derivation
\begin{align*}
& \Rightarrow_{\text{TWFO}}^{\text{AddLiteral}(K)} &&  (MK; O_0; F_0; \top)\\
& \Rightarrow_{\text{TWFO}}^{\text{(Internal Rules)}^*} &&(MK; O; F; \top). \\
\end{align*}
Then in $(MK; O; F; \top)$ all rule applications of the rules CreateInstance, UpdateWatched, FactorizeWatched, DetectPropLiteral, Conflict are only applicable to instances $(C; L_1; L_2) \in O$ which fulfil at least one of the following conditions:
\begin{enumerate}
    \item $\comp(K) = L_1$ or $\comp(K) = L_2$
    \item the instance $(C; L_1; L_2)$ has been generated by an application of CreateInstance, UpdateWatched or FactorizeWatched since the last invocation of AddLiteral($K$).
\end{enumerate} \label{lem:applcheckwatched:afteraddlit}

\item If RemoveLiteral was used to reach the state $(M; O; F; D)$, the rules CreateInstance, UpdateWatched, FactorizeWatched, DetectPropLiteral, Conflict are not applicable. \label{lem:applcheckwatched:afterremove}

\item If AddClause($C$) was applied to reach the state $(M; O; F; \top)$, where now $(C; L_1; L_2)^{\top} \in O$, then CreateInstance, UpdateWatched, FactorizeWatched, Conflict, and DetectPropLiteral are only applicable to the instance $(C; L_1; L_2)$ and newly created instances thereof, i.e., those are the only instances that need to be checked for rule applications. \label{lem:applcheckwatched:afteraddclause}

\end{enumerate}
\end{lemma}

\begin{proof}
    \

    \noindent
    \ref{lem:applcheckwatched:afteraddlit}.
Let $(M'; O'; F'; \top)$ be the state right before application of AddLiteral($K$), and $(MK; O; F; \top)$ be the state after the application of AddLiteral($K$) and an arbitrary number of applications of UpdateWatched, FactorizeWatched, Conflict, CreateInstance, DetectPropLit. Further, assume that the run followed a reasonable strategy.

    Assume an application of any of the rules UpdateWatched, FactorizeWatched, Conflict, CreateInstance, DetectPropLit to an instance $(C; L_1; L_2)$. We prove by contradiction that either $(C; L_1; L_2) \not\in O'$ or $\comp(K) = L_1$ or $\comp(K) = L_2$. Thus, assume the negation holds, i.e., $(C; L_1; L_2) \in O'$ and $\comp(K) \not= L_1$, $\comp(K) \not= L_2$. Then, the respective rule application could have already happened to state $(M'; O'; F'; \top)$, as the instance was already generated, and all other preconditions except conditions on $\comp(L_1) \in M$ do not change when applying any of the rules. However, this contradicts reasonableness.

    \medskip\noindent\ref{lem:applcheckwatched:afterremove}.
    Let $(M; O; F; D)$ be a state reached during a run with a reasonable strategy, where the last applied rule was RemoveLiteral. In this case, it must hold that $D \not= \top$. This rules out any application of UpdateWatched, FactorizeWatched, Conflict, CreateInstance, DetectPropLit.

    \medskip\noindent\ref{lem:applcheckwatched:afteraddclause}.
    Assume that AddClause($C$) was applied to reach the state $(M; O; F; \top)$ in a reasonable run, with $(C; L_1; L_2)^{\top} \in O$. Note that a reasonable derivation that ends in such a scenario is always preceded by a Conflict step, followed by an arbitrary number of RemoveLiteral steps. Note that there must be at least one RemoveLiteral step, as otherwise the previous conflict would still exist, which violates criterion \ref{crit:addclause:confresolved} of AddClause.
	Let $(M'; O'; F; \top)$ be the step right before the rule Conflict is applied. Note that, by definition, $F$ does not change, $O = O' \cup \{(C; L_1; L_2)\}$, and $M \subsetneq M'$ as RemoveLiteral was applied at least once.

    The last step, AddClause($C$), has set the conflict clause $D$ to $\top$, potentially enabling new rule applications. For the sake of a proof by contradiction, assume there is a rule application of CreateInstance, UpdateWatched, FactorizeWatched, Conflict, or DetectPropLiteral that is not to the instance $(C; L_1; L_2)$. Given the definition of $O$, this means that the rule also could be applied to $(M; O'; F; \top)$.

    The rule Conflict cannot be applied to $(M; O'; F; \top)$. Conflict was applied to $(M'; O'; F; \top)$, where $M' \supsetneq M$. Thus, $M'$ was created by AddLiteral, which would not have been applicable in a reasonable run if Conflict would have been.

    The applied rules also cannot be FactorizeWatched, CreateInstance or DetectPropLiteral. If they could be applied to $(M; O'; F; \top)$, by $M \subseteq M'$ they could also be applied to $(M; O; F; \top)$. As Conflict was applied to $(M; O; F; \top)$, this would contradict reasonableness. It remains to show that UpdateWatched cannot become applicable. The rule could become applicable due to two reasons:

    (i) By considering a literal $L_j \not\in M'$, but $L_j \in M$, the literal $L_j$ could become unassigned. Thus, it might be required to watch this unassigned literal instead of a false one. However, this case cannot happen, as by \Cref{lem:watchedInvariants}, if $L_j$ was watched, it was rightmost in $M$. As RemoveLiteral can only remove literals in order, the watched false literal must become unassigned first. So no rule invocation is required.

    (ii) The rule could be applicable if a true literal is removed from $M$, giving the need to update watched literals. Without loss of generality, assume that $\comp(L_1) \in M'$ and $L_2 \in M'$, but $L_2 \not\in M$. Then, either $L_1$ is unassigned in $M'$, in which case UpdateWatched is not applicable, or $\comp(L_1) \in M'$. In the second case, by \Cref{lem:watchedInvariants}, it must hold that $L_2$ is the only unassigned literal. Thus, applying UpdateWatched is not possible.

    The proof for newly created instantiations of $(C; L_1; L_2)$ is analogous to the proof in \ref{lem:applcheckwatched:afteraddlit}.
\end{proof}

\noindent
The following \Cref{lem:efficientApplication} captures the essential properties that make the two-watched literal scheme efficient.
Intuitively, it states that after any incremental problem change (i.e. the application of a steering rule), the set of applicable internal rules is heavily restricted. This allows to omit checks for impossible rules in an implementation.
For example, after an application of AddLiteral($L$), only clause instances that have a literal unifiable with the complementary literal of $L$ watched need to be considered.
This mechanism is key for the efficiency of the two-watched literal scheme: All clause instances that \emph{do not watch} such a literal can be completely ignored at this point.

\begin{restatable}[Two-Watched Literal Scheme]{lemma}{lemmaTwoWatched}
\label{lem:efficientApplication}
  In a reasonable run from an empty start state, the update rules can only be applied if the following holds:
	\begin{enumerate}
		\item After AddLiteral($L$), only clause instances with a watched literal unifiable with $\comp(L)$, and newly created instances thereof, need to be considered for updates.
		\item After AddClause($C$), only the newly added clause, and newly created instances thereof, need to be considered for updates.
		\item After RemoveLiteral, no update rules are applicable.
	\end{enumerate}
\end{restatable}

\begin{proof}
	Corollary of \Cref{lem:applcheckwatched}.
\end{proof}

\noindent

\section{Evaluation}\label{section:evaluation}

We have implemented both schemes presented in \Cref{sec:naive,sec:calculus} as part of the SPASS(SCL) prover \cite{BSW23SCL,CASC30}. First, we give some comments on implementation details and continue with analyzing the performance of both approaches.

\subsection{Implementation}

The two-watched literal scheme is implemented as follows. An instance consists of a clause and a substitution. Both the clause and
the codomain terms in the substitution are shared. Two instances are equal if the clause and the substitution are equal. This may
lead to some superfluous instances but is far more efficient than checking the actual instantiated clauses for equality.
The watched literals as well as the trail are stored in a standard discrimination tree index and path index, respectively~\cite{McCune92a,NieuwenhuisHRV01}.

Now if a new literal on the trail appears, potential watched literals are selected from the watched literal index by
a query searching for generalizations. Then the instances are processed, i.e., watched literals may be exchanged, new
instances may be created including factoring, and propagating literals and conflicts are detected. The data structures
are updated accordingly. Duplicate instances are detected via two-level hashing, where the first level key is the
clause and the second level key the substitution.

\subsection{Experimental Evaluation}
In the following, we compare the two-watched literal scheme for first-order logic (TWFO) with the baseline dynamic programming (DP) approach. To this end, we run SPASS(SCL) on each non-equality TPTP v9.2.1 \cite{Sut17TPTP} problem. For a given problem, both approaches are run concurrently. To our knowledge, there are no other comparable approaches published that can be used for benchmarking.

\subsubsection{Runtime}
To compare the runtimes of both approaches, the prover is run with an overall time limit of 60 seconds.
The time limit includes the runtime of both algorithms, as well as all other steps performed by SPASS(SCL).
In the following, for each problem, we plot the \emph{share} of the total runtime that was spent in the respective approach. By design, the sum of both shares cannot exceed $60$ seconds.
We ensure at runtime that both algorithms are cut off at the same point, i.e., they have performed the exact same computations and yielded equivalent results.

\begin{figure}[H]
\centering
\begin{minipage}{0.48\textwidth}
	\centering
\includegraphics[height=17em,alt={Runtime of TWFO plotted against runtime of DP on $3459$ problems. There are significantly more points below the main diagonal starting from $0.01\text{s}$, i.e., TWFO outperforms DP on most larger problems, while DP outperforms TWFO on very small problems.}]{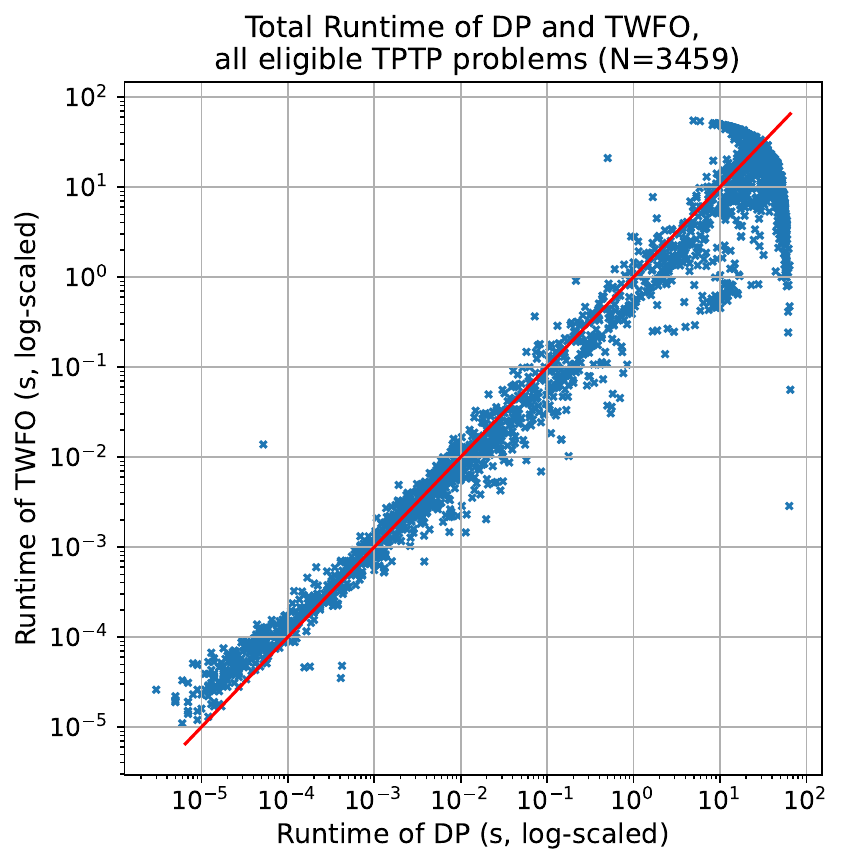}
	\captionof{figure}{Runtime comparison on all eligible TPTP v9.2.1 problems.}
	\label{fig:total-runtime-all}
\end{minipage}\hfill
\begin{minipage}{0.48\textwidth}
	\centering
\includegraphics[height=17em,alt={Runtime of TWFO plotted against runtime of DP on $820$ problems. This time, the majority of the points is below the main diagonal, especially for larger problems.}]{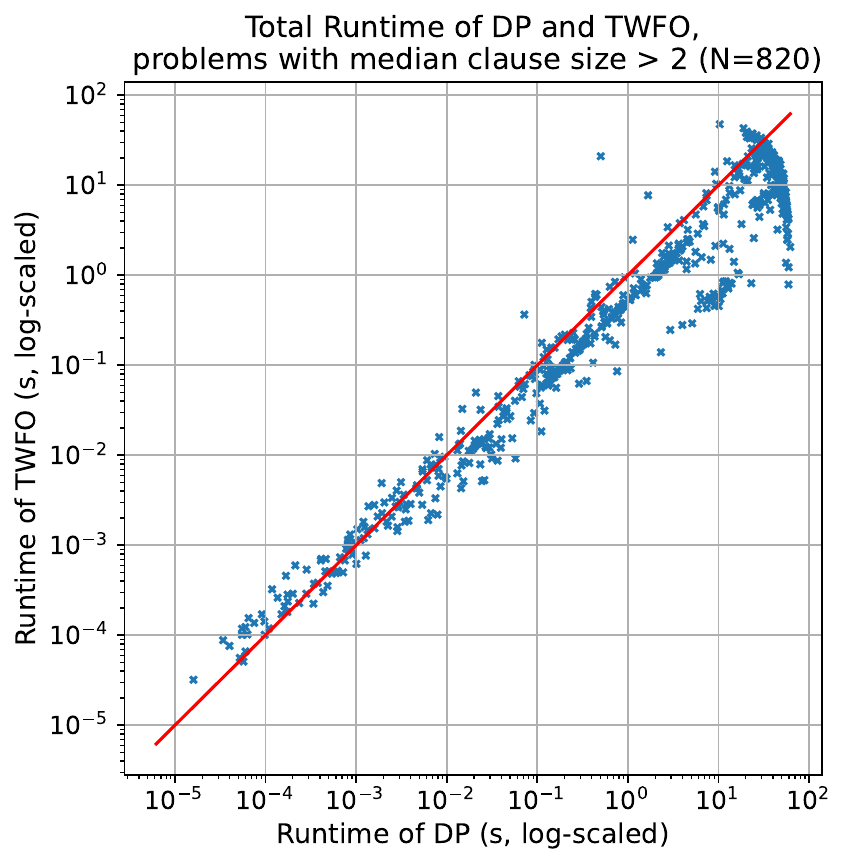}
	\captionof{figure}{Runtime on problems with median clause length $>2$.}
	\label{fig:total-runtime-len2}
\end{minipage}

\vspace{0.5em}

\begin{minipage}{0.48\textwidth}
	\centering
\includegraphics[height=17em,alt={Runtime of TWFO plotted against runtime of DP on $95$ problems. All points are belong the main diagonal, most of them far below.}]{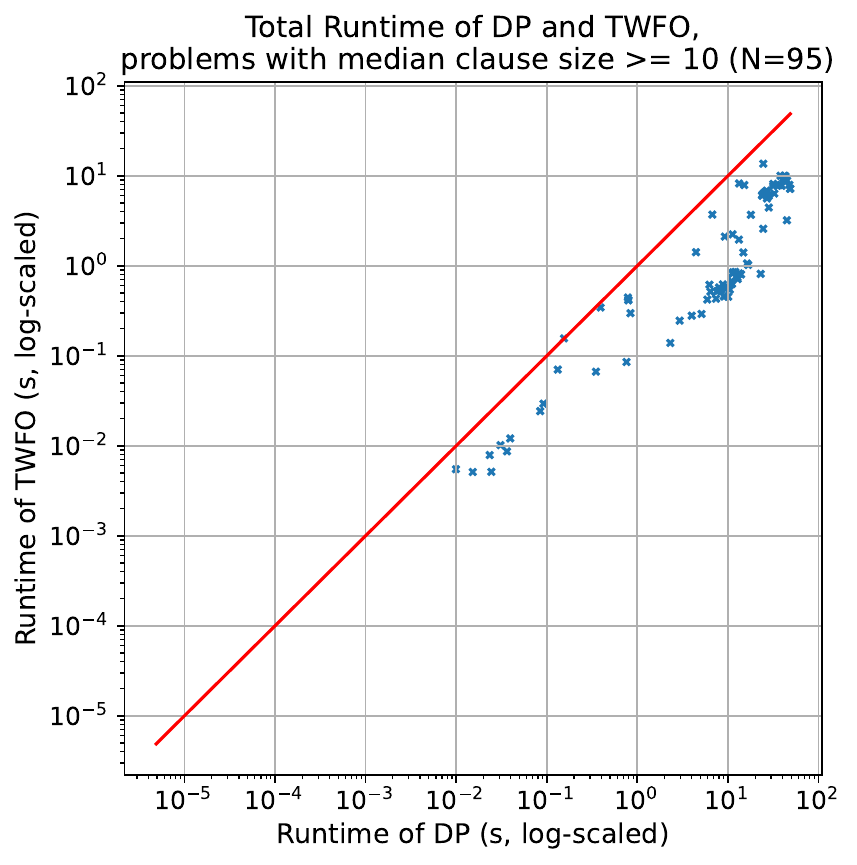}
	\captionof{figure}{Runtime on problems with median clause length $\ge 10$.}
	\label{fig:total-runtime-len10}
\end{minipage}\hfill
\begin{minipage}{0.48\textwidth}
	\centering
\includegraphics[height=17em,alt={Instances generated by TWFO plotted against instances generated by DP on $3459$ problems. There are only very few points above the main diagonal, many points are far below the main diagonal.}]{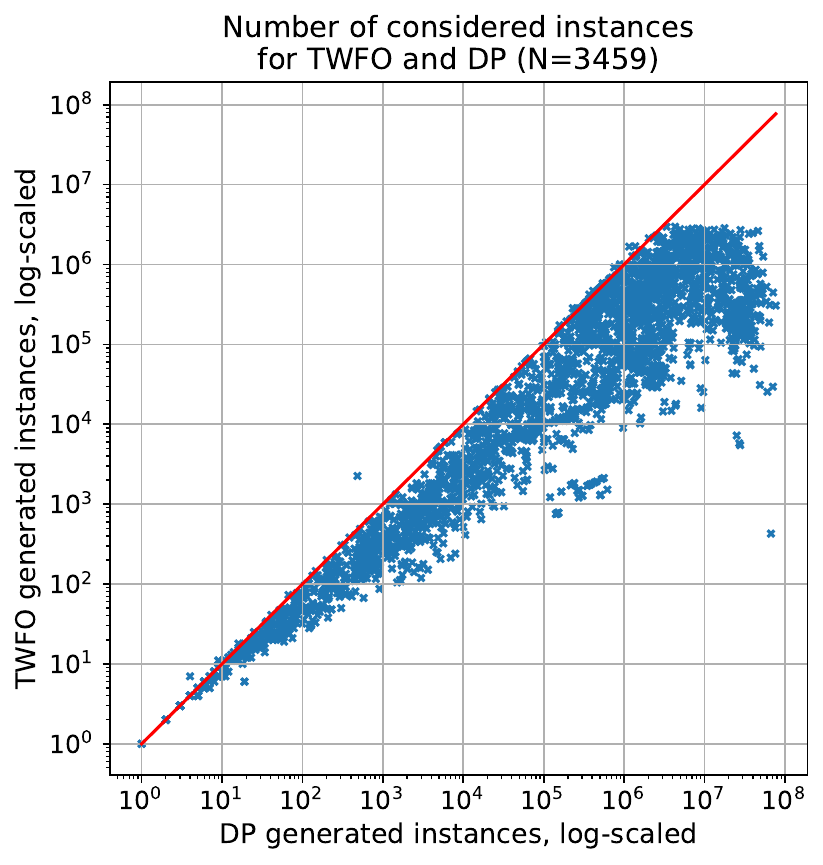}
	\captionof{figure}{Number of considered instances on all eligible problems.}
	\label{fig:instances}
\end{minipage}
\end{figure}

\Cref{fig:total-runtime-all} shows the runtime share of TWFO plotted against the DP approach. The DP approach, in general, outperforms TWFO only on very easy problems that are solved very fast. This stems from the initialization overhead of the TWFO data structures. This changes for larger problems and TWFO often outperforms the baseline DP algorithm, sometimes by orders of magnitude. In the average case, TWFO is faster if the DP approach would spend more than roughly $10\text{ms}$ of runtime with propagation and conflict detection. 
Larger instances in which no significant speedup, or even a slowdown occurs, are often characterized by the presence of many (or exclusively) short clauses. In such cases, a two-watched literal scheme structurally cannot improve performance.

In \Cref{fig:total-runtime-len2}, only problems that have a median clause size greater than $2$ are shown. This excludes problems where, in the majority of cases, improvements cannot be achieved by a two-watched literal scheme. Overall, for such problems, there is a clear trend showing that TWFO outperforms the DP algorithm on the vast majority of larger problems.
\Cref{fig:total-runtime-len10} shows only problems with predominantly large clauses, i.e., clauses that have a median size of at least $10$. Note that the TWFO approach outperforms the DP algorithm in every single problem of this class, very often on the order of a magnitude. The difference in achieved speedup is explained by the number of factors of a clause. TWFO achieves the greatest speedup if the problem clauses do not have a lot of factors. This result is not surprising, as such clauses are similar to propositional clauses, for which two-watched literal schemes are extremely efficient.

Overall, TWFO performs best in the presence of long clauses without many factors. For long clauses, it always outperforms the baseline algorithm, often yielding significant speedups. As clause length can be trivially checked, overall, an implementation which handles clauses differently depending on their length (i.e., long clauses are handled by TWFO, small clauses are handled by the DP approach) seems most promising. 
Similar implementation details for short clauses exist in modern propositional SAT solvers \cite{sorensson2005minisat}.

\subsubsection{Considered Instances}
Another metric we use to compare the approaches is the number of instances that need to be considered during each run.
In the TWFO calculus, instances are easily counted by $|O|$ in a final state $(M; O; F; D)$.
For the baseline dynamic programming approach, we count the number of instances it generates as follows:
Every initial or learned clause directly counts as one instance.
If the clause non-trivially factors to a unit clause, this instance is counted as well.
Last, a new instance is counted whenever a previously unseen substitution is computed.

\Cref{fig:instances} shows that the number of considered instances is almost always smaller in TWFO, and often significantly so. Note that the number of created instances forms a theoretical lower bound on the runtime, abstracting away from the concrete implementation details. This metric clearly shows that, from a theoretical standpoint, the two-watched literal scheme can benefit from significantly fewer considered instances.

\section{Conclusion}\label{section:conclusion}

We have presented two algorithms for first-order conflict and propagation detection: A baseline dynamic programming approach, and a two-watched literal scheme.
Our experimental evaluation shows that the two-watched literal scheme almost always outperforms the baseline approach when measuring the number of instances of a clause that need to be considered. For a total account to propagatable literals extra instances cannot be prevented as shown by the example in the introduction.

Our experiments show that especially on bigger instances with long clauses, the two-watched literal scheme outperforms the baseline approach, often in the order of a magnitude. Thus, implementation of a two-watched literal scheme proves essential for such problems. Overall, our results show that an optimal implementation could default to our dynamic-programming approach for short clauses, while using the proposed two-watched literal scheme for longer clauses.

Although we evaluated our two-watched literal scheme only in a ground trail setting, it should also be successfully applicable in a general
first-order setting, for example when implementing hyperresolution~\cite{RobinsonVoronkov01}. Basically, in this setting we replace instantiation by unification, because
partner literals to the watched clause may have variables, and instead of queries to a ground trail, we query the applicable literals of respective first-order clauses.
For example, in case of a superposition setting, the maximal literals of respective clauses. Depending on the result of the hyperresolution rule, e.g., only positive literals
of Horn clauses are considered, the number of created instances can be further restricted.

\bibliographystyle{splncs04}
\bibliography{twfo}

\end{document}